\documentclass[a4paper,USenglish,cleveref, autoref, thm-restate]{lipics-v2021}
\pdfoutput=1 
\hideLIPIcs  
\nolinenumbers
\graphicspath{{./img/}}
\usepackage[utf8]{inputenc}
\usepackage[T1]{fontenc}
\usepackage{todonotes}
\usepackage{mathtools}
\let\tod\todo
\renewcommand\todo[1]{\tod[inline]{#1}}

\title{Conway Normal Form: Bridging Approaches for Comprehensive Formalization of Surreal Numbers}
\titlerunning{Conway Normal Form for the Formalization of Surreal Numbers }

\author{Karol P\k{a}k}{University of Bia\l{}ystok, Poland}{pakkarol@uwb.edu.pl}{https://orcid.org/0000-0002-7099-1669}{}

\author{Cezary Kaliszyk}{
  University of Melbourne, Australia  and University of Innsbruck, Austria
}{cezarykaliszyk@gmail.com}{https://orcid.org/0000-0002-8273-6059}{}

\authorrunning{K. P\k{a}k and C. Kaliszyk}
\Copyright{Karol Pąk and Cezary Kaliszyk}



\supplement{Formalization can be found at:}
\supplementdetails{Formalization}{http://cl-informatik.uibk.ac.at/cek/itp2024/}
\funding{Supported by ERC PoC grant no.~101156734 \emph{FormalWeb3} and Cost action CA20111 EPN}

\EventEditors{John Q. Open and Joan R. Access}
\EventNoEds{2}
\EventLongTitle{42nd Conference on Very Important Topics (CVIT 2016)}
\EventShortTitle{CVIT 2016}
\EventAcronym{CVIT}
\EventYear{2016}
\EventDate{December 24--27, 2016}
\EventLocation{Little Whinging, United Kingdom}
\EventLogo{}
\SeriesVolume{42}
\ArticleNo{23}

\lstdefinelanguage{Mizar}%
{columns=fullflexible,
keywords={scheme,schemes,environ,provided,where,ranks,
theorem,definition,radix,reserve,properties,struct,inhabited,expandable,attribute,
adjective,registration,coherence,defpred,cluster,from,sch,%
given,such,that,%
reflexivity, irreflexivity, symmetry, asymmetry, connectedness,and,attr,%
    antonym,existence,uniqueness,commutativity,idempotence,synonym,notation,%
    mode, means,func,pred, pred,equals,it,of,is,axiomatization,sethood,reconsider,redefine,%
    if,otherwise,proof,for,ex,being,holds,def,let,consider,take,not,contradiction,st,the,%
    be,thus,implies,assume,then,not,by,or,hence,thesis,end,iff,;,:,",",\#,as,qua},%
   sensitive=true,%
   basicstyle={{\linespread{1.}\usefont{T1}{lmss}{m}{n}}},%
 keywordstyle={\usefont{T1}{lmss}{sbc}{n}\selectfont},%
 keywordstyle=[2]{\it},%
  mathescape = true,%
   morecomment=[l][\texttt]{::},%
   literate={'}{{$\strut\mkern6mu\strut$}}1%
   {&}{{\usefont{T1}{lmss}{sbc}{n}\selectfont\texttt{\&}}}1%
   {-}{{\usefont{T1}{lmss}{m}{n}\selectfont\texttt{-}}}1%
   {->}{{$\rightarrow$}}1%
   {BrLeft}{{$\{$}}1%
   {\{\}}{{\{\!\}}}1
  }

\makeatletter
\def\miz{\lstset{language=Mizar}\lstinline}
\def\mizV[#1]{\lstset{language=Mizar,keywords=[2]{#1}}\lstinline}
\makeatother
\lstnewenvironment{Mizar}[1]
{\lstset{language=Mizar,keywords=[2]{#1},basicstyle=\fontsize{7.6pt}{8.6pt}\selectfont{}}}{}

\DeclareMathAlphabet{\mathbbmsl}{U}{bbm}{m}{sl}
\newcommand\No[0]{{\mathbbmsl{N}\!o}}
\newcommand\On[0]{{\mathbbmsl{O}\!n}}

\newcommand{\Day}[2][]{\mbox{Day}_{#1}\,{#2}}
\newcommand{\born}[2][]{\mathfrak{b}
_{{#1}\,}{#2}}
\newcommand{\Uniq}[1]{\mathfrak{U}_{\scalebox{.7}{$\scriptstyle{}niq$}}{#1}}
\newcommand{\NoOrd}[1]{{\No^{\scalebox{.5}{$\scriptstyle{}\mkern-3mu\leq$}}{#1}}}
\newcommand{\NoOrdInd}[1]{{\No^{^{\scalebox{.5}{$\scriptstyle{}\mkern-3mu\leq$}}}{#1}}}
\newcommand{\Games}[1]{\mbox{Games}\,{#1}}

\newcommand{\Nomega}{\boldsymbol{\omega}}

\usepackage{pict2e}

\DeclareRobustCommand{\KPtriangle}{%
  \begingroup
 \setlength{\unitlength}{8pt}%
  \begin{picture}(1,1)
  \polyline(1,0)(0,0)(0,1)(1,0)
  \end{picture}%
  \endgroup
}

\usepackage{tikz}
\usetikzlibrary{positioning}
\usetikzlibrary{decorations.text}
\usetikzlibrary{decorations.pathmorphing,shapes.multipart}
\usetikzlibrary{shapes,mindmap,trees,decorations.pathreplacing,decorations.pathmorphing,decorations.markings,shapes,matrix,positioning,fit,calc,patterns}

\begin{document}

\maketitle

\begin{abstract}
  The proper class of Conway's surreal numbers
  forms a rich totally ordered algebraically closed field
  with many arithmetic and algebraic properties close to those of
  real numbers, the ordinals, and infinitesimal numbers.
In this paper, we formalize the construction of Conway's
numbers in Mizar using two approaches and propose a bridge between them, aiming
to combine their advantages for efficient
formalization. By replacing transfinite induction-recursion with
transfinite induction, we streamline their construction. Additionally,
we introduce a method to merge proofs from both approaches using global choice,
facilitating formal proof.
We demonstrate that
surreal numbers form a field, including the square root, and that they encompass
subsets such as reals, ordinals, and powers of $\omega$.
We combined Conway's work with Ehrlich's generalization to formally prove Conway's
Normal Form, paving the way for many formal developments in surreal number theory.
\end{abstract}

\section{Introduction}\label{s:intro}

Surreal Numbers, developed by John Conway \cite{Conway}, are fascinating
for several reasons. Their construction relies on two intuitively
simple recursive definitions: as new numbers are built, an order
relation on the numbers is extended. They form a totally ordered
algebraically closed field, denoted by $\No$ that contains (up to
isomorphism) the reals, the ordinals, infinitesimal numbers, as well
as great ones, for example $\frac{1}{\omega}$, $\sqrt{\omega}$,
$\omega^{\omega^{\cdot^{\cdot^\omega}}}$. Despite their intuitively simple
definition, they form a proper class of numbers inductively while
simultaneously defining an order on the class recursively.
In the seventies, when they were first considered, mathematics was not
ready for such definitions, considering them unsure or even
unsafe. Even Conway \cite{Conway} wondered if such definitions were
meaningful and later referred to their construction as ``remarkable''.

Fortunately, the informal concept intrigued several mathematicians and
gave rise to work on the foundations of such concepts
\cite{Ehrlich2012,grimm2012,tondering}. Conway saw the possibility of
defining a model for surreal numbers in Neumann-Bernays-Gödel (NBG)
set theory with global choice. This has later been completed by
Ehrlich \cite{Ehrlich}.  There are also several more detailed proofs
of the properties of surreal numbers as a field; however, the more
formal ones only cover their simpler properties
\cite{alabdullah,grimm2012,tondering}.  Among these, Schleicher's work
\cite{schleicher} has been essential for our formalization.

There are two approaches to defining surreal numbers. The first
approach is closer to Conway's convention: it begins by considering the quotient
of surreal numbers with respect to the equivalence
relation
(defined by $x \le y \land y \le x$) and proceeds
to demonstrate that it forms a pre-order, employing switching between
representatives of the equivalence classes. On the other hand, the
second approach defines the class of surreal numbers using a
tree-theoretic definition
\cite{Ehrlich}. Both approaches have their advantages and
disadvantages, with the former allowing for the free selection of
equivalence class representatives, while the latter ensures the
uniqueness of elements. This means that fomalization of certain proofs
is more intricate and challenging in one approach compared to the other.

In this paper, we define a bridge that allows us to combine formal
proofs about surreal number in the general sense with the
tree-theoretic proofs. This allows us to efficiently formalize a large
number of properties of surreal numbers. In particular:
\begin{itemize}
\item We propose an easier approach to constructing the general surreal numbers where transfinite induction-recursion is replaced by transfinite induction.
\item We propose an approach to conveniently combine the general approach proofs with tree-theoretic approach using global choice.

\item We show that this is a convenient representation by proving that the surreals form a field including the square root.
\item We show that reals, ordinals, and powers of $\omega$ are subsets of the surreals. 
\item We combine the Conway Normal Form proof skeleton \cite{Conway}
with Erlich's generalization of Conway's theory of surreal numbers
\cite{Ehrlich} expressed in NBG set theory to prove it in Tarski-Groethendieck set theory~\cite{BrownPak19}. 
\item We present the details of the formalization in the Tarski-Grothendieck set theory formalized in the Mizar proof system. 
 The same approach can be used to work effectively with surreal numbers in systems that only support transfinite induction, e.g. Isabelle/ZF~\cite{isabellezf} and MetaMath~\cite{metamath}.
\item We proved a large number of surreal number properties needed for the above results. This amounts to
  335 proved top-level Mizar theorems totaling 1099 KB. The parts of the formalization
  corresponding to the proofs that surreals form a ring is already in the Mizar library
  \cite{SURREAL0.ABS,SURREALO.ABS,SURREALR.ABS}.
\end{itemize}
To our knowledge, this is the most in-depth formalization of surreal numbers today.

\section{Mizar}

The Mizar proof system operates within the framework of classical
first-order logic, augmented with limited second-order schemes
requiring explicit instantiation by users~\cite{mizar,Grabowski2015FDM}.  Within that logic, Mizar
introduces the axioms of Tarski-Grothendieck set theory~\cite{MML2017}, which extends
ZFC by incorporating Tarski's Axiom A. This axiom implies the Axiom of
Choice (AC)~\cite{BrownPak19} and enables the existence of arbitrarily large strongly
inaccessible cardinals, thus providing models of ZFC and circumventing
the necessity for proper classes in certain formalizations.

Unlike traditional foundational type systems, Mizar treats types as
first-order predicates, supplemented with automation implemented
through user-programmable Horn clauses. These clauses facilitate the
propagation of various properties, known as adjectives, throughout
Mizar terms in a bottom-up fashion \cite{JAR2018}. The fact that an object
\mizV[x]{X} satisfies the type predicate \mizV[t]{t} is written \miz{x be t}.
Mizar adopts a Jaśkowski-style natural deduction approach,
complemented by a fast and type-aware refutational first-order prover
with several extensions \cite{DBLP:conf/itp/GrabowskiK23}
known as the Mizar obvious inference \miz{by}.

 The Mizar system is
accompanied by the Mizar Mathematical Library (MML), a large corpus of
formal mathematics, that among many other topics includes
formalizations of reals and ordinals that we will use in this work.

Mizar allows the definition of several kinds of meta-level objects. We
briefly revisit the definition of meta-level functions by \miz{means},
as we will utilize them several times throughout the paper. This
mechanism enables the definition of a function based on a
predicate. The predicate can then reference the object being defined
using the special keyword \miz{it} within the definition body. This
methodology closely resembles defining functions using the choice
operator found in other proof systems. However, unlike the direct use
of the choice operator, function definitions by \miz{means} in Mizar permit
the specification of additional conditions and require the explicit
declaration of the result type.
Two key proof obligations accompany such definitions: the proof of
existence and the proof of uniqueness of the function's result. We
illustrate the syntax of function definitions by \miz{means} in Mizar
through an example.
For two sets \mizV[X]{X} and \mizV[Y]{Y} of type \miz{set}, the MML
defines a meta-level function (keyword \miz{func}) union, denoted as
\mizV[X,Y]{X'$\cup$'Y}. This function returns a set (return type
indicated after the \miz{$\to$} keyword).
The semantics of the function
(following the \miz{means} keyword) are given by the predicate
stating that elements belong to the union if they are in any of its
arguments. While the Mizar input syntax for universal and existential
quantifiers is \miz{for} and \miz{ex}, respectively, we present
them using more standard quantifier symbols in this paper:

\begin{Mizar}{X,Y,x}
''let'X,'Y'be'set;
''func'X'$\cup$'Y'$\to$'set'means
''''$\forall$x'be'set.'x'$\in$'it'$\Leftrightarrow$'(x'$\in$'X'$\lor$'x'$\in$'Y);
\end{Mizar}
Each meta-function definition  requires 
showing that the defined object exists
and is unique. For the details of these proofs see the formalization.

Meta-level predicates and types (as mentioned above types are just predicates) are defined
in an analogous way, with the only exception that the keyword \miz{func} is replaced by
\miz{pred} and \miz{attr} respectively.

One of the restrictions we will consider in this paper is that, in Mizar (as well as in similar systems based on set theory),
everything is considered to be a set, in accordance with Tarski's first axiom.
Additionally, any set is an element of a set in the von Neumann hierarchy of sets.
However, this does not imply that reasoning about certain classes is impossible,
as meta-level functions and predicates (in Mizar also attributes/types) are not sets.
This means that, according to the grammar, \mizV[x]{x'is'set} can only be written when \mizV[x]{x} is a term.
Consequently, we can define a type such as \miz{Ordinal} even if there is no set of ordinals.
Similarly, we can define meta-level functions on a type whose elements form a proper class,
for example, a successor function of the type \miz{Ordinal} $\to$ \miz{Ordinal}.
In Mizar (as in Isabelle/ZF or Megalodon), we can even quantify over such meta-level functions
and predicates using second-order logic, but only with the universal quantifier.

\section{Introduction to Surreal Numbers}\label{s:surreal}
\newcommand\restr[2]{{\left.\kern-\nulldelimiterspace {#1}
  {\mathchoice{\vphantom{\big|}}{}{}{}}\right|_{#2}
  }}
Conway introduced the surreal numbers using two interleaving definitions: The way to build a new surreal number relies on two sets
of surreal numbers, for which appropriate ordering constraints hold. And the way to check if two numbers are related in the order relies
on checking the relation for the underlying sets of numbers. More precisely:
\begin{itemize}
\item[]$\strut\mkern-27mu${\bf Concept:} If $L$, $R$ are any two sets of numbers, and no member of $R$ is $\leq$ any member
of $L$, then there is a number $\{ L \mid R \}$. All numbers are constructed in this
way.
\item[]$\strut\mkern-26mu${\bf Comparison:}
If $x=\{ L\mid R\}$, $x^\prime=\{ L^\prime \mid R^\prime\}$, then $x \leq x^\prime$ if and only if
$x^\prime \not\leq$ any member of $L$
and no member of $R$ is $\leq x^\prime$.
\end{itemize}

We introduce several notations that allow describing the various properties and proofs more concisely.
We write $L \ll R$ iff for each $x \in L$ and $y \in R$, $y \not\leq x$.
We introduce the relation $x \thickapprox y$ to denote $x \le y \land y \le x$. It is easy to see that it is an equivalence relation.
Note, that several works quoted in the introduction Sec.~\ref{s:intro} use the same symbol for identity $=$ and equivalence $\thickapprox$ of
surreal numbers. As we aim to formalize these using interactive proof systems we will precisely separate the two.
 We also follow Conway's original notation
 $\{L\mid R\}$ for pairs $\langle L,R\rangle$. Let $x=\{ L\mid R\}$ a surreal number, we use $L_x$ and $R_x$ to refer
 to the left $L$ and the right component $R$ of $x$, respectively.


Conway constructs surreal numbers in so-called days indexed by ordinals. He starts by defining the first number, denoted ${0_\No }$, as the pair $\{\,\mid\,\}$ ($\coloneqq \langle \emptyset,\emptyset\rangle$). Note that the relation $\emptyset \ll \emptyset$ obviously holds. This number is then used to initialize day zero as the only number present at this stage.
%
In the next iteration, $\Day{1}$, we could consider three more pairs
$-1_\No\coloneqq\{\,\mid\! {0_\No} \}$, ${1_\No}\coloneqq\{{0_\No}\mid\,\}$, and $\{{0_\No}\mid {0_\No}\}$. The last one of those is not a number since already in $\Day{0}$ we can prove that
${0_\No}\leq {0_\No}$, and relations between numbers are preserved across days.
Generally, $\Day{\alpha}$ is defined by all  surreal numbers $x$ for which $L_x,R_x\subseteq \bigcup_{\beta<\alpha} \Day{\beta}$ and $L_x\ll R_x$.
Additionally, we introduce the concept of the birthday of a surreal number $x$, denoted by $\born{x}$, i.e,
the smallest ordinal $\alpha$ such that $x \in \Day{\alpha}$.

To show the main differences between Conway's approach and the tree-theoretic one, that we address in our contribution, we present the construction of  $\Day{2}$.
When generating the numbers present in $\Day{2}$, we can place any of the numbers already in $\Day{1}$ (i.e., $-1_\No$, $0_\No$, $1_\No$)
in the left and in the right set. This gives $(2^3)^2 = 64$ candidates for new numbers $x$.
Only 20 of these numbers satisfy the criterion $L_x\ll R_x$.
Note, that checking this criterion requires knowing the ordering on all the numbers in the preceding $\Day{1}$.
In general, this number grows exponentially, that is given $n$ different numbers, there are
$(n+2)2^{n-1}$ new ones that satisfy the comparison criterion.
 These 20 numbers are different, however, not all are different in the quotient $\thickapprox$. There are, in fact, only 4 new numbers, namely
 $-2_\No$, $-\frac{1}{2}_\No$,
 $\frac{1}{2}_\No$, $2_\No$ (see Fig.~\ref{fig:day1_4}).
 This is because some of the newly generated numbers are equivalent in the $\thickapprox$ sense to each other,
 e.g., $\frac{1}{2}_\No\coloneqq\{{0_\No}\,\mid\! {1_\No} \} \thickapprox \{ {-1_\No},{0_\No}\,\mid\! {1_\No}\}$,
 and some are equivalent to already existing ones, e.g.,
 $0_\No=\{\,\mid\, \}
\thickapprox \{ {-1_\No}\mid\,\}
\thickapprox \{\,\mid {1_\No}\}
\thickapprox \{{-1_\No}\mid {1_\No}\}$.
More generally, in $\Day{n}$ for any $n\in\mathbb{N}$ there are $2n$ new numbers.

We next need to characterize the comparison relation on the new numbers. To do this for the new numbers
in $\Day{\alpha}$, it is not sufficient to directly (without recursion) use the order for the previous
days $\Day{\beta}$ for $\beta<\alpha$. As the tree has depth $\alpha$, we need to
perform up to $\alpha$ steps of recursion in order to use the information from previous days.
Indeed, to justify that  $\{{-1_\No}\mid {1_\No}\} \leq \{{0_\No}\,\mid\! {1_\No} \}$ (compare with $0_\No \leq \frac{1}{2}_\No$)
we need to show ${-1_\No}\leq \{{0_\No}\,\mid\! {1_\No} \}\wedge \{{-1_\No}\mid {1_\No}\}\leq 1_\No$, and that happens because
$-1_\No \leq 1_\No$ which we know from $\Day{1}$.
 \begin{figure}
\begin{center}
\newcommand{\Blarg}{\fontsize{6.8pt}{7.3pt}\selectfont{}}
\newcommand{\Mlarg}{\fontsize{6.5pt}{7.2pt}\selectfont{}}
\newcommand{\Bnorm}{\fontsize{6.4pt}{7.1pt}\selectfont{}}
\newcommand{\Mnorm}{\fontsize{5.4pt}{6.0pt}\selectfont{}}
\newcommand{\Bsmall}{\fontsize{5.4pt}{6.0pt}\selectfont{}}
\newcommand{\Msmall}{\fontsize{4.0pt}{4.8pt}\selectfont{}}
\begin{tikzpicture}
   \matrix[%
        column sep={24pt,between origins},
        row sep={14pt,between origins},
        punkt/.style={rectangle, rounded corners, draw,line width=0.1pt, minimum size=41pt,rectangle split,
            rectangle  split  parts=2
            }] (before)
   { & &&&&&&
   \node[punkt](A0){{ $0_\No$}\nodepart{second}$\{\;\;\mid\;\;\}$};
& &&&&&&\\ \\
&&&\node[punkt](AM1){\Blarg{ $-1_\No$}\nodepart{second}\Mlarg$\{ \;\; \mid 0_\No \}$};
&&&&&&&&
\node[punkt](A1){\Blarg { $1_\No$}\nodepart{second}\Mlarg$\{0_\No \mid \;\; \}$};
\\
\\
\\
   &
\node[punkt](AM2){\Bnorm$\strut-2_\No$\nodepart{second}\Mnorm$\{\;\;\strut\mid\! {-1_\No} \}$};&&&&
\node[punkt](AM12){\Bnorm$-\frac{1}{2}_\No$\nodepart{second}\Mnorm$\{{-1_\No}\,\mid\! {0_\No} \}$}; &&&&
\node[punkt](A12){\Bnorm$\frac{1}{2}_\No$\nodepart{second} \Mnorm$\{{0_\No}\,\mid\! {1_\No} \}$}; &&&&
\node[punkt](A2){\Bnorm$\strut2_\No$\nodepart{second}\Mnorm$\{{1_\No}\!\mid \;\; \}$};&
\\
\\
\\
\node[punkt](AM3){\Bsmall$\strut-3_\No$\nodepart{second}\Msmall$\{\;\;\strut\mid\! {\strut-2_\No} \}$};&&
\node[punkt](AM32){\Bsmall$\strut-\frac{3}{2}_\No$\nodepart{second} \Msmall$\{{-\!2_\No}\!\mid\! {\strut-1_\No} \}$}; &&
\node[punkt](AM34){\Bsmall$\strut-\frac{3}{4}_\No$\nodepart{second} \Msmall$\{{-\!1_\No}\!\mid\! {\strut\!-\frac{1}{2}_\No} \}$}; &&
\node[punkt](AM14){\Bsmall$\strut-\frac{1}{4}_\No$\nodepart{second}\Msmall$\{\strut{-\!\frac{1}{2}_\No}\!\mid \! {0_\No} \}$};&&
\node[punkt](A14){\Bsmall$\strut\frac{1}{4}_\No$\nodepart{second}\Msmall$\{ \strut{0_\No}  \!\mid \! {\frac{1}{2}_\No} \}$};&&
\node[punkt](A34){\Bsmall$\strut\frac{3}{4}_\No$\nodepart{second}\Msmall $\{\strut{\frac{1}{2}_\No}\!\mid\! {1_\No} \}$};&&
\node[punkt](A32){\Bsmall$\strut\frac{3}{2}_\No$\nodepart{second}\Msmall $\{\strut{1_\No}\!\mid\! {2_\No} \}$}; &&
\node[punkt](A3){\Bsmall$\strut3_\No$\nodepart{second}\Msmall $\{\strut{2_\No}\!\mid\;\; \}$};\\
   };
\draw(A0)--(AM1);
\draw(A0)--(A1);
\draw(AM1)--(AM2);\draw(AM1)--(AM12);
\draw(AM2)--(AM3);\draw(AM2)--(AM32);\draw(AM12)--(AM34);\draw(AM12)--(AM14);
\draw(A1)--(A2);\draw(A1)--(A12);
\draw(A2)--(A32);\draw(A2)--(A3);
\draw(A12)--(A34);\draw(A12)--(A14);

\draw(AM34)--($(AM34)+(-0.5,-0.8)$);\draw(AM34)--($(AM34)+(0.5,-0.8)$);
\draw(AM14)--($(AM14)+(-0.5,-0.8)$);\draw(AM14)--($(AM14)+(0.5,-0.8)$);
\draw(AM32)--($(AM32)+(-0.5,-0.8)$);\draw(AM32)--($(AM32)+(0.5,-0.8)$);
\draw(AM3)--($(AM3)+(-0.5,-0.8)$); \draw(AM3)--($(AM3)+(0.5,-0.8)$);

\draw(A34)--($(A34)+(-0.5,-0.8)$);\draw(A34)--($(A34)+(0.5,-0.8)$);
\draw(A14)--($(A14)+(-0.5,-0.8)$);\draw(A14)--($(A14)+(0.5,-0.8)$);
\draw(A32)--($(A32)+(-0.5,-0.8)$);\draw(A32)--($(A32)+(0.5,-0.8)$);
\draw(A3)--($(A3)+(-0.5,-0.8)$); \draw(A3)--($(A3)+(0.5,-0.8)$);

\node at ($0.5*(A0)+0.5*(AM1)+(0,0.2)$) {\pmb{-}};
\node at ($0.5*(A0)+0.5*(A1)+(0,0.2)$) {\pmb{+}};

\node at ($0.5*(AM1)+0.5*(AM2)+(-0.3,0.1)$) {\pmb{-}};
\node at ($0.5*(AM1)+0.5*(AM12)+(0.3,0.1)$) {\pmb{+}};
\node at ($0.5*(A1)+0.5*(A2)+(0.3,0.1)$) {\pmb{+}};
\node at ($0.5*(A1)+0.5*(A12)+(-0.3,0.1)$) {\pmb{-}};

\node at ($0.5*(AM2)+0.5*(AM3)+(-0.3,0)$) {\pmb{-}};
\node at ($0.5*(AM2)+0.5*(AM32)+(0.3,0)$) {\pmb{+}};
\node at ($0.5*(AM12)+0.5*(AM34)+(-0.3,0)$) {\pmb{-}};
\node at ($0.5*(AM12)+0.5*(AM14)+(0.3,0)$) {\pmb{+}};
\node at ($0.5*(A12)+0.5*(A14)+(-0.3,0)$) {\pmb{-}};
\node at ($0.5*(A12)+0.5*(A34)+(0.3,0)$) {\pmb{+}};

\node at ($0.5*(A2)+0.5*(A32)+(-0.3,0)$) {\pmb{-}};
\node at ($0.5*(A2)+0.5*(A3)+(0.3,0)$) {\pmb{+}};

\node at ($(AM3)+(-0.65,-0.7)$) {{\small\pmb{-}}};
\node at ($(AM3)+(0.7,-0.7)$) {{\small\pmb{+}}};
\node at ($(AM32)+(-0.65,-0.7)$) {{\small\pmb{-}}};
\node at ($(AM32)+(0.7,-0.7)$) {{\small\pmb{+}}};
\node at ($(AM34)+(-0.65,-0.7)$) {{\small\pmb{-}}};
\node at ($(AM34)+(0.7,-0.7)$) {{\small\pmb{+}}};
\node at ($(AM14)+(-0.65,-0.7)$) {{\small\pmb{-}}};
\node at ($(AM14)+(0.7,-0.7)$) {{\small\pmb{+}}};
\node at ($(A14)+(-0.65,-0.7)$) {{\small\pmb{-}}};
\node at ($(A14)+(0.7,-0.7)$) {{\small\pmb{+}}};
\node at ($(A34)+(-0.65,-0.7)$) {{\small\pmb{-}}};
\node at ($(A34)+(0.7,-0.7)$) {{\small\pmb{+}}};
\node at ($(A32)+(-0.65,-0.7)$) {{\small\pmb{-}}};
\node at ($(A32)+(0.7,-0.7)$) {{\small\pmb{+}}};
\node at ($(A3)+(-0.65,-0.7)$) {{\small\pmb{-}}};
\node at ($(A3)+(0.7,-0.7)$) {{\small\pmb{+}}};
   \end{tikzpicture}
\end{center}
\caption{The relations between the numbers created in the first four days. The bottom part of each node give a representative of the equivalence class w.r.t. the $\thickapprox$ relation.
  The edge labels \pmb{+} and \pmb{-} correspond to the tree-theoretic interpretation of surreal numbers.}
  \label{fig:day1_4}
\end{figure}
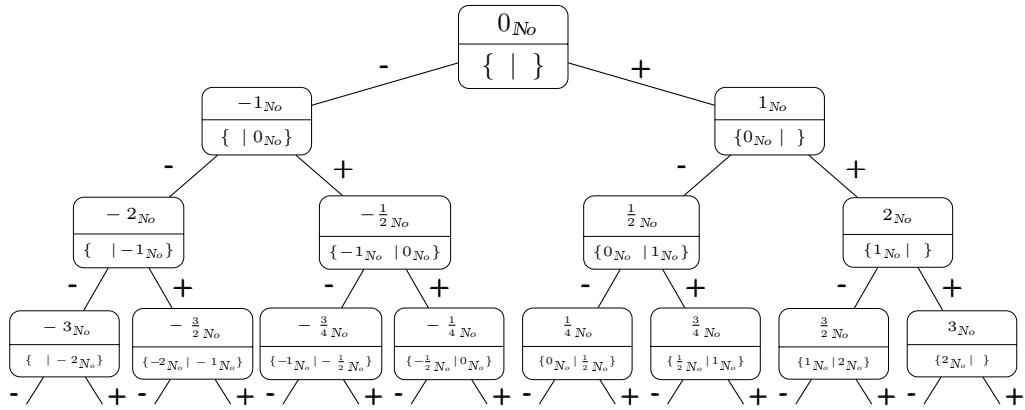

For any surreal number $x$ in Fig.~\ref{fig:day1_4}, we can always construct two new
numbers in the $\thickapprox$ sense using $x_1=\{L_x \cup \{x\} \mid R_x\}$ and
$x_2=\{L_x \mid R_x \cup \{x\}\}$ to represent them. Naturally,
if $x$ is youngest w.r.t. $\thickapprox$, then $\born{x_1}=\born{x_2}=1\!+\!\born{x}$.
Note, that the new numbers in $\Day{\alpha}$, where $\alpha$ is a limit ordinal,
cannot have a direct predecessor. They are created as cuts, similar to Dedekind reals.
Nevertheless,
if $\born{x}$ is not a limit ordinal and $x$ is youngest w.r.t. $\thickapprox$,
we can construct $y$ that corresponds to the direct predecessor of $x$ in the $\thickapprox$ sense, for which
$x\thickapprox\{L_y \cup \{y\} \mid R_y\}$ or
$x\thickapprox\{L_y \mid R_y \cup \{y\}\}$.
This allows interpreting the
equivalence classes of the $\thickapprox$ relation, as the class of all
possible maps from an ordinal (including limit ordinals) to the set $\{\pmb{+},\pmb{-}\}$, where \pmb{+} and $\pmb{-}$ correspond to these two
alternatives.
This is the foundation of the tree-theoretic approach.
As ordinals can be thought of as sequences, we can use the standard lexicographic order, with
$\pmb{-}\prec\mbox{\emph{undefined}}\prec \pmb{+}$. As the sequences are of different
length for different ordinals, the \emph{undefined} come up outside of the domain of the maps:
If $x$ is a subsequence of length $\alpha$ of the sequence $y$ then the first index where they
differ is $\alpha+1$. At that index $x(\alpha+1)$ is undefined while $y(\alpha+1)$ has a value.

In the tree-theoretic approach, the comparison is defined as below. We will
not analyse it in our work, but it helps compare the approaches.
\begin{definition}[tree-theoretic comparison]\label{lex}
Let $x,y$ be maps from an ordinal to $\{\pmb{+},\pmb{-}\}$,
that represent surreal numbers in the tree-theoretic approach.
 Suppose that $x \neq y$ and $\alpha$ is  the smallest ordinal where $x(\alpha) \neq y(\alpha)$.
 Then $x< y$ if and only if
\begin{equation}
\small
\left(x(\alpha)=\pmb{-}\wedge y(\alpha)\mbox{ is \emph{undefined}}\right)\vee
\left(x(\alpha)=\pmb{-}\wedge y(\alpha)=\pmb{+}\right)\vee
\left(x(\alpha)\mbox{ is \emph{undefined}}\wedge y(\alpha)=\pmb{+}\right).
\end{equation}
\end{definition}
The construction of surreal numbers in this approach is significantly easier to formalize \cite{CGAMES_1.ABS}.
We can even define the negation operator $-x$ by exchanging \pmb{+} and \pmb{-} in any map $x$.
However, the construction of the remaining field operations becomes much more involved. In order to define $x\star y$, it is
needed to build a kind of bridge to the Conway approach. This starts with some representation $\overline{x},\overline{y}$, for which
we define $\overline{x}\star\overline{y}$ using the Conway method. Finally, one needs to show the existence of a map for $x\star y$, rather than
use direct recursion. A similar approach is required to prove all the field properties. This is much more involved than
in the Conway approach, where we can freely represent numbers by their different representatives in their $\thickapprox$ class.

\section{Formal Set-theoretic Construction of Surreal Numbers}\label{s:constr}
In the previous Section,
we pointed out that
the surreal numbers are not a set and
their ordering relation $\leq$
cannot be a set.
However, the restriction of this relation
to any particular day is a set.
To work with such sets, we will index the order $\alpha$ using a relation $Ord$ that is a set.
The notation $x \leq_{Ord} y$ simply means that
$\langle x,y\rangle \in {Ord}$ and $L \ll_{Ord} R$ means $\forall\,l\!\in\!L.\:\forall\,r\!\in\!R.\: \langle r,l\rangle \not\in Ord$.
Remember, that constructing the surreal numbers in $\Day{2}$ proceeded in two steps: First the candidates were selected using the ordering on $\Day{1}$ surreals; subsequently the order in $\Day{2}$ was computed. In the construction, as well as in the uniqueness proof, we need to modify
the  ${Ord}$ relations with a given set of candidates to construct $\Day{\alpha}$. For this, we define the sets of pairs
 $\Games{\alpha}$ for any ordinal $\alpha$ as follows ($\mathcal{P}$ stands for powerset):
\begin{equation}
\Games{\alpha} = \mathcal{P}\left(\displaystyle\bigcup_{\beta<\alpha} \Games{\beta}\right)\times
\mathcal{P}\left(\displaystyle\bigcup_{\beta<\alpha} \Games{\beta}\right).
\end{equation}
Clearly,  $\Games{0} = \{ \emptyset \} \times \{ \emptyset \} = \{ \langle \emptyset,\emptyset\rangle \}$
and $\Day{\alpha} \subseteq \Games{\alpha}$.
Now we can define $\Day[Ord]{\alpha}$
even if $Ord$  does not satisfy the {\it Comparison} condition as follows:
\begin{equation}
\Day[Ord]{\alpha} = \{ x \in \Games{\alpha}\mid L_x \subseteq \displaystyle\bigcup_{\beta<\alpha} \Day[Ord]{\beta} \wedge
  R_x \subseteq\displaystyle\bigcup_{\beta<\alpha} \Day[Ord]{\beta} \wedge L_x \ll_{Ord} R_x\}.
\end{equation}
The recursive definition relies on a very complicated recursion scheme that combines unions over all previous ordinals.
As this is very hard to express in several systems, including Mizar, we use a helper sequence ${S}$.
\begin{definition}[{$\Day[Ord]{\alpha}$}]\label{def:dayord}
Let $Ord$ be relation, $\alpha$ be an ordinal and a $\alpha$-length sequence $S$ that satisfies:
\begin{equation}
 S(\beta) = \{ x \in \Games{\beta}\mid L_x \subseteq \displaystyle\bigcup_{\gamma<\beta} S(\gamma) \wedge
  R_x \subseteq \displaystyle\bigcup_{\gamma<\beta} S(\gamma) \wedge L_x \ll_{Ord} R_x\}
\end{equation}
for any ordinal $\beta\leq \alpha$. Then $\Day[Ord]{\alpha}=S(\alpha)$.
\end{definition}

We give the formal definition of $\Day[Ord]{\alpha}$ in Mizar and explain several used concepts below:

\begin{Mizar}{A,B,S,x,R,Ord}
'let'$\alpha$'be'Ordinal,'Ord'be'Relation;
'func'Day(Ord,$\alpha$)'->'Subset'of'Games'$\alpha$'means
'''$\exists\,$S'be'Sequence.'it'='S.$\alpha$'$\land$'dom'S'='succ $\alpha$'$\land$'$(\forall\,\beta$'be'Ordinal.'$\beta$'$\in$'succ'$\alpha$'$\Rightarrow$
'''''S.$\beta$'='{x'where'x'is'Element'of'Games'$\beta$:'L$_x$'$\subseteq$'union'rng'($\restr{S}{\beta}$)'$\land$'R$_x$'$\subseteq$'union'rng'($\restr{S}{\beta}$)'$\land$'L$_x$'$\ll_{Ord}$'R$_x$});
\end{Mizar}
\begin{itemize}
\item The length of \mizV[S]{S} is $\alpha$ (equivalently \mizV[S]{dom'S'='succ'$\alpha$}),
\item \mizV[S]{S.$\beta$} is the set-theoretic function application, corresponding to $\Day[Ord]{\beta}$
for \mizV[S]{$\beta$'$\in$'succ'$\alpha$} ($\beta < \alpha$).
\item \mizV[S]{union'rng'($\restr{S}{\beta}$)} is the union of the values of the sequence \mizV[S]{S} restricted to the ordinal $\beta$
which corresponds to \mizV[Ord]{$\bigcup_{\gamma<\beta} \Day[Ord]{\gamma}$},
\item $it$ refers to the defined object, equal to \mizV[S]{S.$\alpha$} which also is a subset of \mizV[]{Games'$\alpha$}.
\end{itemize}

The definition implies that $\Day[Ord]\alpha\subseteq \Day[Ord]\beta$ if $\alpha\leq \beta$,
but it is also possible to use transfinite induction to show a ``monotonicity''-like property:

\begin{lemma}\label{lem:mono}
Let $Ord$ be a relation, $\alpha$ be an ordinal and $x\in\Games{\alpha}$
such that
$x\not\in \Day[Ord]{\alpha}$. Then for all ordinals $\beta$, $x\not\in \Day[Ord]{\beta}$.
\end{lemma}

We also restrict the concept of birthday of a surreal number $x$ to a relation $Ord$.

\begin{definition}
Let $Ord$ be relation, $x$ be element of $\Day[Ord]{\beta}$ for some ordinal $\beta$.
Then the birthday of a object $x$ with respect to a relation $Ord$
is an ordinal $\alpha$ that satisfies two conditions:
\begin{itemize}
\item $o \in \Day[Ord] \alpha$,
\item $\forall\,\beta.\:\left( x \in \Day[Ord] \beta \rightarrow \alpha \leq \beta\right)$.
\end{itemize}

\begin{Mizar}{A,B,S,x,R,Ord}
'assume'$\exists\,\beta$'be'Ordinal.'x'$\in$ Day(Ord,$\beta$);
'func'born(Ord,x)'->'Ordinal'means
'''x'$\in$'Day(Ord,it)'$\wedge$'$(\forall\,\beta$'be'Ordinal.'x'$\in$'Day(Ord,$\beta$)'$\Rightarrow$'it'$\subseteq$'$\beta$);
\end{Mizar}
\end{definition}

Even if we construct the  $Ord$ relation that satisfies the {\it Comparison} condition, we cannot use that single
relation for all days. Indeed, each day $\Day[Ord]{\alpha}$ corresponds to at least one new number
$\{\bigcup_{\beta<\alpha} \Day[Ord]{\beta}\mid\,\}$ and $\langle{0_\No},\{\bigcup_{\beta<\alpha} \Day[Ord]{\beta}\mid\,\}\rangle\in Ord$.
As such, we can only assume that  $Ord$ satisfies the {\it Comparison} condition on some set.
At this point, we could already talk about the surreal numbers, but only until a certain birthday,
after which we would have to modify the relation.

Conway, uses a special induction over $n$-tuples of arguments (referred to as Conway's induction)
when constructing the surreal numbers, their order, the operations, as well as when proving their properties.
The induction is similar to $\in$-induction, where the truth of $P(x_1,x_2,\ldots,x_n)$ follows
from the truth of $P$ for all modified tuples $x_1,x_2,\ldots,x_n$ where at least one
$x_i$ is replaced by its left or right component.
Unfortunately, some proofs require changing the order of the elements in a sequence with additional properties.
This has already been observed by Mamane \cite{Mamane04}: In his Coq formalization he refers to such induction arguments as {\it permuting inductions}.
Such an induction could be expressed as a transfinite induction over sums of $\born{}$ applied to these arguments,
however, standard ordinal sum is not symmetric nor monotonous on its both arguments.

In our formalization, we tackle this problem by using the natural Hessenberg sum of ordinals.
This is a variant of ordinal sum that is symmetric and monotonic on both arguments.
In fact, Hessenberg invented his ordinal sum, inspired by surreal numbers and the use of Cantor Normal Form\footnote{
Cantor Normal Form is the unique representation of any ordinal $x$ as \(n_1\omega^{\alpha_1}+n_2\omega^{\alpha_2}+\ldots+n_k\omega^{\alpha_k}\)
for some $k\in\mathbb{N}$, where $\{n_i\}_{i=0}^k$ is a sequence of positive naturals
and $\{\alpha_i\}_{i=0}^k$ a decreasing ordinal sequence.
},
but it became a concept in mathematics in general and has already been formalized \cite{ORDINAL7.ABS}.
In our work, wherever possible, we use subsets of Cartesian product in the construction,
but full Hessenberg sum is necessary for example in the definitions of the arithmetic operation.

\newcommand{\CProd}[1]{\mbox{\texttt{Prod}}^{\,C}_{\,#1}}
\newcommand{\OProd}[1]{\mbox{\texttt{Prod}}^{\,O}_{\,#1}}
\begin{definition}[$\CProd{}$ and $\OProd{}$]\label{defCOprod}
Let  $Ord$ be a relation and $\alpha$, $\beta$ be ordinals. Then we define two subsets of the Cartesian product
$\Day[Ord]{\alpha}\times\Day[Ord]{\alpha}$ as follows:
\begin{align}
\CProd{Ord}(\alpha,\beta) & = \{\langle x,y\rangle\mid x,y\in \Day[Ord]{\alpha}\wedge\left((\born[Ord]{\,x}<\alpha\wedge \born[Ord]{\,y}<\alpha)\vee\right.&\notag\\
     & \strut\mkern19mu\left.(\born[Ord]{\,x}=\alpha\wedge \born[Ord]{\,y}\leq\beta) \vee (\born[Ord]{\,x}\leq\alpha\wedge \born[Ord]{\,y}=\beta) \right) \}&\\
\OProd{Ord}(\alpha,\beta) & = \{\langle x,y\rangle\mid x,y\in \Day[Ord]{\alpha}\wedge\left((\born[Ord]{\,x}<\alpha\wedge \born[Ord]{\,y}<\alpha)\vee\right.&\notag\\
     & \strut\mkern19mu\left.(\born[Ord]{\,x}=\alpha\wedge \born[Ord]{\,y}<\beta) \vee (\born[Ord]{\,x}<\alpha\wedge \born[Ord]{\,y}=\beta) \right) \}&
\end{align}
\end{definition}

The $^O$ and $^C$ superscripts are used, since the concepts are somewhat similar to open and closed intervals respectively.
Observe two properties of $\CProd{}$ and $\OProd{}$: $\CProd{Ord}(\alpha,\alpha)=\Day[Ord]{\alpha}\times\Day[Ord]{\alpha}$ and
$\bigcup_{\gamma<\beta}\CProd{Ord}(\alpha,\gamma) = \OProd{Ord}(\alpha,\beta)$.

As we already discussed, the order relation on surreals $\leq$ is too big to be a set, so
reasoning about it is complicated. For this reason, we will consider its subsets that
are sets. A restriction of the $\leq$ relation to any set, will be a subset of
$\Day{\alpha}\times \Day{\alpha}$ for some $\alpha$,
so we introduce a natural restriction:

\begin{definition}[Almost $\No$-order]
A relation $Ord$ is an \emph{almost $\No$-order} if $Ord\subseteq \Day[Ord]{\alpha}\times \Day[Ord]{\alpha}$
for some ordinal $\alpha$.
\end{definition}

\newcommand{\Comp}[2]{\mbox{Comp}({#1},{#2})}
\begin{definition}
Let $A$ be a set. A relation $Ord$ preserves the Comparison condition on $A$ (written $\Comp{Ord}{A}$)
if and only if
\begin{equation}
\forall\,x.\:\forall\,y.\:x \leq_{Ord} y \Leftrightarrow L_x \ll_{Ord} \{y\} \wedge \{x\}\ll_{Ord} \{y\}.
\end{equation}
\end{definition}

One of the crucial properties of our formalization is that we can complete the proofs using only (transfinite) induction, without
requiring any complex techniques available only in selected systems (e.g. we do not use induction-recursion or complicated recursion schemes).
Apart from the construction we here
show the first proof done this way in full.
The following Theorem \ref{ThP2subclaim} gives a form of uniqueness of the order on the surreal numbers, uses
two applications of (transfinite) induction. Many proofs in our formalization use two inductions in a similar way.

\begin{theorem}\label{Th:uniq}
Let $R$,$S$ be relation. The following facts hold.
\begin{enumerate}
\item if $R\,\cap \bigcup_{\gamma<\alpha} \Games \gamma = S\,\cap \bigcup_{\gamma<\alpha} \Games \gamma$
 where $\alpha$ is any ordinal, then:
 \begin{itemize}
 \item $\Day(R,\alpha)=\Day(S,\alpha)$,
 \item $\forall\,x\!\in\!\Day(R,\alpha).\:\born[R](a) = \born[S](a)$,
  \item $\forall\,\beta.\: R \cap \OProd{R}(\alpha,\beta) = S \cap \OProd{S}(\alpha,\beta)\wedge
 R \cap \CProd{R}(\alpha,\beta) = S \cap \CProd{S}(\alpha,\beta)$.
\end{itemize}
\item if $R,S$ are almost $\No$-order, $\Comp{R}{\CProd{R}(\alpha,\beta)}$, $\Comp{S}{\CProd{S}(\alpha,\beta)}$
then $R\cap\CProd{R}(\alpha,\beta) = S\cap \CProd{S}(\alpha,\beta)$.
\end{enumerate}
\end{theorem}
\begin{proof}
Fact 1 is proved by transfinite induction.
Let $x=\{L_x\mid R_x\} \in \Day[R]{\alpha}$. Every element of $L_x \cup R_x$ is a memeber of $\Day[R]{\beta}$
for some $\beta<\alpha$, so by induction hypothesis it is also a member of $\Day[S]{\beta}$.
Similarly, $L_x \ll_S R_x$ (equivalently $\forall\,l\!\in\!L_x.\:\forall\,r\!\in\!R_x.\:l \not\leq_{S} r$)
 is a consequence of $L_x \ll_R R_x$ and
$R\,\cap \bigcup_{\gamma<\alpha} \Games \gamma = S\,\cap \bigcup_{\gamma<\alpha} \Games \gamma$,
therefore $x\in \Day[S]{\alpha}$. Since $\forall\,\beta\!\leq\!\alpha.\:\Day[R]{\beta} = \Day[S]{\beta}$
the remaining part of Fact 1 is straightforward.

To show Fact 2, first consider the following subclaim. If $R,S$ are almost $\No$-orders, then
\begin{multline}\label{ThP2subclaim}
\forall\,\alpha.\:\forall\,\beta.\:
\beta\leq \alpha\wedge \Comp{R}{\CProd{R}(\alpha,\beta)}\wedge\Comp{S}{\CProd{S}(\alpha,\beta)}\wedge\\
\strut\mkern43mu R \cap \OProd{R}(\alpha,\beta) = S \cap \OProd{S}(\alpha,\beta) \Rightarrow \\
R \cap \left(\CProd{R}(\alpha,\beta)\setminus \OProd{R}(\alpha,\beta)\right)
\subseteq
S \cap \left(\CProd{S}(\alpha,\beta)\setminus \OProd{S}(\alpha,\beta)\right)
\end{multline}
Let $\beta,\alpha$ such that
$\beta\leq \alpha$, $\Comp{R}{\CProd{R}(\alpha,\beta)}$, $\Comp{S}{\CProd{S}(\alpha,\beta)}$
and
\begin{equation}\label{ThP2eq1}
R\cap \OProd{R}(\alpha,\beta) = S \cap \OProd{S}(\alpha,\beta).
\end{equation}
Let $x,y$ such that $\langle x,y\rangle \in R \cap \left(\CProd{R}(\alpha,\beta)\setminus \OProd{R}(\alpha,\beta)\right)$.
Then, by definition \ref{defCOprod}, there are only two possible cases:
$\born[R]x = \alpha\wedge\born[R]y = \beta$ or $\born[R]x = \beta\wedge\born[R]y = \alpha$.
Without loss of generality we assume that $\born[R]x = \alpha\wedge\born[R]y = \beta$.
By Lemma \ref{lem:mono} we know that
$R\,\cap\, \bigcup_{\gamma<\alpha} \Games \gamma
= S\,\cap\, \bigcup_{\gamma<\alpha} \Games \gamma$,
hence $\langle x,y\rangle \in \CProd{S}(\alpha,\beta)\setminus\OProd{S}(\alpha,\beta)$, $\born[S]x = \alpha$, $\born[S]y = \beta$.
It remains to prove $\langle x,y\rangle \in S$. Since $\Comp{S}{\CProd{S}(\alpha,\beta)}$, this is equivalent
to $L_x\ll_S\{y\}\wedge\{x\}\ll_S R_y$.

To prove the first conjunct, suppose contrary to our claim, that $y \leq_S l$ for some $l\in L_x$.
Then either $\beta<\alpha$ or $\beta=\alpha$. In both cases
we get $\langle y,l\rangle\in S\,\cap\,\OProd{S}(\alpha,\beta)$
since $\born[S]l <\born[S]x$.
Hence $y \leq_R l$ by \eqref{ThP2eq1} contrary to $L_x\ll_R\{y\}$ (by $\langle x,y\rangle\in R$).
To show the secound conjunct, $\{x\}\ll_S R_y$,
 suppose
contrary to our claim, that $r \leq_S x$ for some $r\in R_y$.
Then $\born[S]r <\born[S]y$, $\langle r,x\rangle\in S\,\cap\,\OProd{S}(\alpha,\beta)$,
and finally $r \leq_R x$ by \eqref{ThP2eq1} contradicting ${x}\ll_R\{y\}$ (by $\langle x,y\rangle\in R$).

We can now easily infer $R \cap \CProd{R}(\alpha,\beta) = S \cap \CProd{S}(\alpha,\beta)$ from
$R \cap \OProd{R}(\alpha,\beta) = S \cap \OProd{S}(\alpha,\beta)$.
Next, using transfinite induction we can simplify the assumption, obtaining
\begin{multline}\label{ThP2subclaim2}
\forall\,\alpha.\:\forall\,\beta.\:
\beta\leq \alpha\wedge \Comp{R}{\CProd{R}(\alpha,\beta)}\wedge\Comp{S}{\CProd{S}(\alpha,\beta)}\wedge\\
 R \cap \OProd{R}(\alpha,0) = S \cap \OProd{S}(\alpha,0) \Rightarrow
R \cap \CProd{R}(\alpha,\beta) =
S \cap \CProd{S}(\alpha,\beta)
\end{multline}
with the help of the equation $\bigcup_{\gamma<\beta}\CProd{Ord}(\alpha,\gamma) = \OProd{Ord}(\alpha,\beta)$.
A second use of transfinite induction together with the equation
$\bigcup_{\beta<\alpha}\CProd{Ord}(\beta,\beta) = \OProd{Ord}(\alpha,0)$
completes the proof of fact 2.
\end{proof}

The above Theorem \ref{ThP2subclaim} allows defining the order using
its selected properties. We only show the main step in the construction, the
remaining part are two applications of transfininte induction, analogously to
what we did in the proof of the second part of the lemma \ref{ThP2subclaim}.

\begin{theorem}\label{Th:ex}
Let $\alpha,\beta$ ordinals, $R$ be relation such that
$\Comp{R}{\OProd{R}(\alpha,\beta)}$ and $R \subseteq \OProd{R}(\alpha,\beta)$.
Then
\begin{multline}
S \coloneqq R\,\cup\,\{\langle x,y\rangle\mid x,y \in \Day[R]{\alpha} \wedge
\left((\born[R]{x}= \alpha \wedge \born[R]{y}= \beta) \vee\right.\\
\left.(\born[R]{x}= \beta \wedge \born[R]{y}= \alpha)\right)\wedge
L_x \ll_{R} \{y\} \wedge \{x\} \ll_{R} R_y\}
\end{multline}
satisfies $\Comp{S}{\CProd{S}(\alpha,\beta)}$ and $S \subseteq \CProd{S}(\alpha,\beta)$.
\end{theorem}

We can now define the order relation. Proving that it is a function (its existence and uniqueness)
relies on the properties of
 $\Comp{Ord}{\CProd{Ord}(\alpha,\alpha)}$ and $Ord\subseteq \CProd{Ord}(\alpha,\alpha)$:
\begin{Mizar}{}
  let'$\alpha$'be'Ordinal;
  func'$\NoOrd{\alpha}$'->'Relation'means'::'SURREALO:def 12
    it'preserves_$\No$_Comparison_on'[:Day(it,$\alpha$),Day(it,$\alpha$):]'$\land$'it'$\subseteq$'[:Day(it,$\alpha$),Day(it,$\alpha$):];
\end{Mizar}

For a given $\alpha$, this relation
is still a set with all usual restrictions of sets.
However, in the formalization we can consider
a different relation for each particular day.
We define the type \miz{surreal}
as the members of at least one day of the form $\Day[\NoOrdInd{\alpha}]{\alpha}$.
Similarly, we can define the order  $x \leq y$ (as a predicate and not a set-theoretic relation)
 as true when
$\langle x,y\rangle$ is an element of at least one  $\NoOrd{\alpha}$ as follows:\\
\begin{minipage}{.49\textwidth}
\begin{Mizar}{x}
''let'x'be'object;
''attr'x'is'surreal'means
''''$\exists\,\alpha$'be'Ordinal.'x'$\in$'$\Day[\NoOrdInd{\alpha}]{\alpha}$;
\end{Mizar}
\end{minipage}
\begin{minipage}{.49\textwidth}
\begin{Mizar}{x}
''let'x,y'be'surreal'object;
''pred'x'$\leq$'y'means
''''$\exists\,\alpha$'be'Ordinal.'x'$\leq_{\NoOrdInd{\alpha}}$'y;
\end{Mizar}
\end{minipage}\\
This gives us the Concept and Comparison properties,
that encapsulate the constructed surreals.
With just the help of standard (transfinite) induction, we created a theoretical heaven, where
the properties of Conway numbers are satisfied. In the next section, we will also introduce
the canonical representation \cite{Ehrlich} that links this to the tree-theoretic approach.

\section{The Surreal Numbers as a Field}\label{s:field}
For a surreal number  $x$, Conway introduces the somewhat confusing concept of a \emph{typical
  member} of $L_x$ and $R_x$ denoted by $x^L$ and $x^R$, respecively. We will use these only
to define how functions affect the left and right parts of a surreal number.
More formally $f(x^L)$ will denote $\{f(y) \mid y \in L_x\}$ (analogous for $x^R$ in $R_x$).\footnote{
Conway also uses typical members in more ambiguous ways, i.e., $x=\{x^L\mid x^R\}$. We will avoid these in this paper.}

With this notation, the unary negation can be simply written as $-x=\{-x^R\mid-x^L\}$ and unfolds to the full
$-x = \{\{-x_r\mid x_r \in R_x\}\mid \{-x_l\mid x_l \in L_x\}\}$.
We write the definitions of the remaining operations only using the typical member notation: $x+y = \{x^L+y,x+ y^L \mid x^R+y,x+y^R\}$,
$x\cdot y = \{
    x^L\cdot y+x\cdot y^L-x^L\cdot y^L,
    x^R\cdot y+x\cdot y^R-\!x^R\cdot y^R \mid
    x^L\cdot y+x\cdot y^R-\!x^L\cdot y^R,
    x^R\cdot y+x\cdot y^L-\!x^R\cdot y^L\}$
where the comma corresponds to the different possible ways how elements are constructed, formally corresponding to a union.

The definion of any arithmetic operation on the surreal numbers would require some complicated recursion
principle, a different one for each operation. To avoid this, all arithmetic operations defined in
our formalization are introduced in three phases.
To define operation $\star$ (such as $+$, $-$, $\cdot$, $^{-1}$) we first define $\star_\alpha$ on
a specific domain of surreal numbers restricted to $\Day{\alpha}$. We will denote this domain by $X_\alpha$
so that we can uniformly cover unary and binary operations. In particular, $X_\alpha$ will be a subset of
surreal numbers for unary $\star$ and a set of pairs for binary operations.
Subsequently, we prove that the application of the operator $\star_\alpha$ on surreal arguments is also surreal. For this
proof we recursively use the properties of $\star_\beta$ for $\beta<\alpha$
as well as for $\star_\alpha$.
Finally, we can define $\star$ as $\star_\alpha$, where $\alpha$ is determined by the arguments given to $\star$.

The second step of each of the definitions is typically most involved and often requires proving several additional
properties. For example, consider the proof that multiplication
restricted to $X_\alpha$ preserves the surreal type (the definition of multiplication relies on addition, so this is
of course done after addition is defined and its basic properties are proved).
The formal goal is: $\forall\,x.\:\forall\,y.\: \langle x,y\rangle \in X_\alpha \Rightarrow x\cdot_\alpha y \mbox{ is surreal}$.
For this, we only have to show  that $x\cdot_\alpha y$ satisfies the \emph{Concept} condition and (following Conway's work for $\cdot$)
requires the following four properties of $\cdot_\beta$ for any $\beta < \alpha$:
\begin{itemize}
\item $\forall\,x.\:\forall\,y.\: \langle x,y\rangle \in X_\beta \Rightarrow x\cdot_\beta y \mbox{ is surreal}$,
\item $\forall\,x.\:\forall\,y.\:\langle x,y\rangle \in X_\beta \Rightarrow x \cdot_\beta y = y \cdot_\beta x$,
\item $\forall\,x_1.\:\forall\,x_2.\:\forall\,y.\:
\langle x_1,y\rangle \in X_\beta \wedge
\langle x_2,y\rangle \in X_\beta \wedge
x_1 \thickapprox x_2
\Rightarrow x_1\cdot_\beta y_1 \thickapprox x_2\cdot_\beta y_1$,
\item $\forall\,x_1.\:\forall\,x_2.\:\forall\,y_1.\:\forall\,y_2.\:\ldots\:\wedge
x_1 < x_2 \wedge y_1 < y_2 \Rightarrow x_1\cdot_\beta y_2+x_2\cdot_\beta y_1 < x_1\cdot_\beta y_1+x_2\cdot_\beta y_2$%
\footnote{where $\ldots$ is the union of assumptions of the form
$\langle x_i,y_j\rangle \in X_\beta$ for every occurrence of $x_i\cdot_\beta y_j$.},
\end{itemize}
necessary in the proof for $\cdot_\alpha$ (proof by induction over $\beta$ of the conjunction of these properties).

In the final step of each definition, that is to define $\star$ based on $\star_\alpha$, we need
to specify the $X_\alpha$ on which it is defined. For the unary operations $-$, $^{-1}$
(but also the square root $\sqrt{.}$, unique element $\Uniq{.}$ for each class $[x]_\thickapprox$ introduced in Section~\ref{s:subset},
and $\Nomega^{.}$ in Section~\ref{s:cnf})
we can
define $\star_\alpha$ on the set $X_\alpha\coloneqq\Day\alpha$.
For binary operations, we introduce $X_\alpha\coloneqq \KPtriangle\!^\alpha$, i.e.,
the set of pairs $\langle x,y\rangle$ where $\born{x} \oplus \born{y} \leq \alpha$
and $\oplus$ is the natural Hessenberg sum of ordinals.
The ordinals are not a set, so we use the same trick as in Definition \ref{def:dayord}, namely
a monotonously increasing sequence of functions (they need to be monotonous, that is only add new
pairs to the set-theoretic functions in order to preserve consitency).
The sequence $\{\Day\alpha\}$ is naturally increasing, and $\{\KPtriangle\!^\alpha\}$ is monotonously
increasing since $\oplus$ is monotonous.
With this, we can show that the sequence of functions $\star_\alpha$
is increasing w.r.t. set inclusion.
We call this property $\subseteq$-\emph{monotone}.
More precisely, a sequence of functions $S$ is $\subseteq$-\emph{monotone}
iff $\forall\,\alpha\in domain(S).\:\forall\,\beta\in domain(S).\:\beta\leq \alpha \Rightarrow
S(\beta)\subseteq S(\alpha)$.
As a consequence $\bigcup_{\beta<\alpha} \star_\beta$ is a function defined on
$\bigcup_{\beta<\alpha} X_\beta$. It already has the necessary properties, but we still have to show
that the results of the function are surreal numbers. The above shows the main stages of the proposed
approach to defining functions on the surreal numbers. The initial construction does not know the properties
of the results; subsequently these are proved; and finally we show that the result is surreal.

In order to efficiently proceed with the formalization, we prove a second-order theorem (referred to as a scheme in Mizar)
that will allow us to efficiently define all these operations. The scheme constructs a sequence, where
the definition has access to all previous elements.
\begin{theorem}\label{seqexuni}
  Let $X$ be a function from ordinals that returns an arbitrary set, and
  $H$ a binary function (the first argument is an arbitrary set but the second must be
  a $\subseteq$-monotone sequence of functions) such that
\begin{itemize}
\item $\forall\,S\mkern3mu be\mkern2mu\subseteq\!-\emph{monotone sequence of functions}.\,
\left(\forall\,\alpha\!\in\!domain(S).\, domain(S(\alpha))=X(\alpha)\right)\Rightarrow$\\
    $\strut\mkern25mu\left(\forall\,\alpha\in domain(S).\:\forall\,x\in domain(S(\alpha)).\;H(x,S) = H(x,\restr{S}{\alpha})\right)$,
\item $\forall\,\alpha.\:\forall\,\beta.\:\beta \leq \alpha \Rightarrow X(\beta) \subseteq X(\alpha)$.
\end{itemize}
Then for every ordinal $\theta$, there exists a unique $\subseteq$-\emph{monotone} sequence $S$ of functions of length $\theta$
where $\forall\,\alpha \in domain(S).\:  domain(S(\alpha))=X(\alpha)\wedge (\forall\,x\in X(\alpha).\: S(\alpha)(x)\!=\!H(x,S))$.
\end{theorem}
\newcommand{\image}{\ensuremath{^{\diagdown}}}

Recall, that $-x = \{-x^R\mid-x^L\}$, or more precisely
$-_\alpha(x) = \{ (\bigcup_{\beta<\alpha} -_\beta)\image R_x \mid
(\bigcup_{\beta<\alpha} -_\beta) \image L_x\}$,
where the image of a set is denoted by \image.
Thus $-_\alpha(x)$ depends on $x$ and $\bigcup_{\beta<\alpha} -_\beta$.
However, we define $-_\alpha(x)$ as $H(x, \{-_\beta\}_{\beta\leq\alpha})$
to be able to access the whole sequence and not just its union.
Additionally, the constraint $x\in X(\alpha) \Rightarrow H(x,S) = H(x,\restr{S}{\alpha})$
intuitively means that we expect
$-_\alpha(x)  = H(x, \{-_\beta\}_{\beta\leq\alpha}) = H(x, \{-_\beta\}_{\beta\leq\theta})$.

We show the details of the definition of $-x$ and $x\!+\!y$ using this theorem.
For $-x$, we simply use $\theta = \born{x}$, $X(\alpha)=\Day\alpha$,
$H(o,S) =  \langle (\bigcup rng\,S)\image R_{o}, (\bigcup rng\,S)\image L_{o} \rangle$.
Finally, define $-x$ to be $-_{\born{x}} x$.
To define addition $x\!+\!y$,
we use $\theta = \born{x} \oplus \born{y}$, $X(\alpha)= \KPtriangle\!^\alpha$,
$H(o,S) = \langle (\bigcup rng\,S)\image(  (L_{L_o}\!\times\!\{R_{o}\})\cup (\{L_o\}\!\times\!L_{R_o})),
(\bigcup rng\,S)\image((R_{L_o}\!\times\!\{R_o\})\cup (\{L_o\}\!\times\! R_{R_o}))\rangle$
and define $x + y$ as $x +_{\born{x} \oplus \born{y}} y$.
To clarify the second case, notice that $o$ as a member of the triangle operator is a pair $o=\langle a,b\rangle$
for some surreal $a,b$, so
the expression  $H$ can be represented as
$\langle +\image((L_a\! \times\! \{b\}) \cup (\{a\}\!\times\!L_b)), +\image(R_a\!\times\!\{b\})\cup(\{a\}\!\times\! R_b))\rangle$ equal $\{L_a\!+\!b, a\!+\!L_b \mid R_a\!+\!b, a\!+\!R_b\}$.
This is well-defined and reduces to Conway's definition of $a+b$.

After defining all the operations, we show that their results are surreal numbers.
Like Schleicher \cite{schleicher}, we show the properties of the operations alongside
the properties of the order.
 We covered all of Schleicher's chapter 3 \cite{schleicher}, additionally making use of
 some  more detailed proofs found in Grimm \cite{grimm2012} and Tondering \cite{tondering}.
 We finally formally show that $\No$ has all the properties of an ordered field:
\begin{equation}\label{field_axiom}
\begin{array}{r@{\:}c@{\:}lcr@{\:}c@{\:}lcr@{\:}c@{\:}l}
x + y & = & y + x && x + 0_{\No} &=& x && x\cdot(y+z) &\thickapprox& x\cdot y + x\cdot z \\
x\cdot y & = & y\cdot x   && x \cdot 1_{\No} &=& x && (x\cdot y)\cdot z &\thickapprox& x\cdot(y\cdot z)\\
  (x+y)+z & = & x+(y+z) && x + (-x) &\thickapprox& 0_{\No} && \multicolumn{3}{c}{x \not\thickapprox 0_\No \rightarrow x \cdot x^{-1} \thickapprox 1_{\No}}
\end{array}
\end{equation}
The inverse operation deserves special attention. Conway uses changing the equivalence
class representative and hides a secondary recursion.
Let $x > 0_\No$.
Then
$x\thickapprox p \coloneqq \{0_{\No},x^L\mid x^R\}$
where $x^L$, $x^R$ are restricted to positive typical members.
More formally, $p$ is constructed by eliminating the non-positive
elements from $L_x$, $R_y$ and adding $0_\No$ in its left component.
Then Conway defines $y\coloneqq x^{-1}$ using a strongly informal recursive property,
that expresses the transitive closure of subsequent generations of typical members of $y$
as follows:
\begin{equation}\label{inv1}
y = \{{0_{\No}},\frac{{1_{\No}}+(p^R-p)y^L}{p^R},
   \frac{{1_{\No}}+(p^L-p)y^R}{p^L} \mid
   \frac{{1_{\No}}+(p^L-p)y^L}{p^L},
   \frac{{1_{\No}}+(p^R-p)y^R}{p^R}\}
\end{equation}
where $p_L$, $p_R$ are reserved only for positive typical members%
\footnote{To show the intricacy of the definition, consider
  $x=5_\No \thickapprox p\coloneqq \{0_\No,4_\No\mid\,\}$
where there is no $p_R$ and $p_L=4_\No$.
Since $y=\{ 0_{\No},\ldots\mid \ldots\}$ we start by $y_L = 0_\No$
and obtain a new $y_R = \frac{1_\No+(4_\No-5_\No)\cdot 1_\No}{4_\No} =\frac{1}{4}_\No$. Then
$\frac{1_\No+(4_\No-5_\No)\cdot \frac{1}{4}_\No}{4_\No}=\frac{3}{16}_\No$ is a new $y_L$,
and $\frac{1_\No+(4_\No-5_\No)\cdot \frac{3}{16}_\No}{4_\No} =\frac{13}{64}_\No$ is a new $y_R$
and so on, creating $y=\{ 0_{\No},\frac{3}{16}_\No,\frac{51}{256}_\No,\ldots\mid \frac{1}{4}_\No, \frac{13}{64}_\No,\frac{205}{1024}_\No \ldots\}$.}.
We imitate Schleicher and Stoll's proof \cite{schleicher},
formalizing using transitive closures of subsequent right and left closures of typical members. For details, see the \texttt{SURREALI.miz} formalization.

A somewhat similar approach for highly-recursive definitions has been used to define the square root.
According to Conway \cite{Conway}, Clive Bach defined the square root of a non-negative number as follows:
\begin{equation}\label{root}
\sqrt{x}=y = \{\sqrt{x^L},\frac{x+y^L\cdot y^R}{y^L+ y^R} \mid
   \sqrt{x^R},\frac{x+y^L\cdot y^{L\bullet}}{y^{L}+ y^{L\bullet}},
   \frac{x+y^R\cdot y^{R\bullet}}{y^{R}+ y^{R\bullet}}\}
\end{equation}
where $x^L$, $x^R$ are non-negative typical members of $x$, and
$y^L$, $y^{L\bullet}$, $y^R$, $y^{R\bullet}$ are options of $y$ chosen so that no
one of the three denominators is zero. Conway leaves the correctness proof to the reader \cite{Conway}.

\newcommand{\SQRT}[3]{\texttt{S}({#1},{#2},{#3})}

\newcommand{\varL}{\texttt{L}}
\newcommand{\varR}{\texttt{R}}

We formally define this in two steps. First, we designate the initial sets $\varL\coloneqq \{\sqrt{x_L}\mid 0_\No \leq x^L \in L_x \}$,
$\varR\coloneqq \{\sqrt{x^R}\mid 0\leq x^R \in R_x\}$, and then we perform a transitive closure. We focus on the second step. For the square root definition, we introduce a helper definition $S$.
\begin{definition}[Square root]
Let $x$ be a surreal numbers, $\varL,\varR$ be surreal number sets.
We define $\texttt{S}(x,A,B)$ to be
      $\{\frac{x+ a \cdot b}{a+b}\mid a \in A \wedge b \in B \wedge 0_\No \not\thickapprox a+b\}$.
Consider the sequences $\{\varL_n\}_{n\in\mathbb{N}}$,
$\{\varR_n\}_{n\in\mathbb{N}}$ defined as follows:
\begin{equation}
\varL_0 = \varL, \;\; \varR_0 = \varR,\;\;
\varL_{n+1} = \varL_n \cup \SQRT{x}{\varL^n}{\varR_n},\;\;
\varR_{n+1} = \varR_n \cup \SQRT{x}{\varL_n}{\varL_n} \cup \SQRT{x}{\varR_n}{\varR_n}.
\end{equation}
Then, we can define $\sqrt{x,\varL,\varR}\coloneqq\{\bigcup_{n\in\mathbb{N}}\varL_n\mid  \bigcup_{n\in\mathbb{N}}\varR_n\}$.
\end{definition}

With the above definition, we can apply Theorem \ref{seqexuni} using $\theta = \born{x}$, $X(\alpha) =\Day\alpha$,
\begin{equation}\label{defHH}H(o,S)= \sqrt{o, \{(\bigcup rng\,S)(o^L)\mid 0_\No\leq o^L \in L_o \},
\{(\bigcup rng\,S) o^R\mid 0_\No\leq o^R \in R_o\}}\end{equation} to obtain the final
function $\sqrt{x}$.

We have proved the usual properties of the square root, such as $\sqrt{x\cdot x}\thickapprox x$ for non-negative
surreals and $\sqrt{x}\,^{-1}\thickapprox \sqrt{x^{-1}}$ for positive surreals.
The definition proposed by Clive Bach can even be applied even to negative surreal
numbers. As expected, it does not behave well for these: rather than give us the
surcomplex numbers we instead show that different representatives of the same
equivalence class of a negative number give different square roots.
Indeed, consider an arbitrary positive $x$. Then
$-1_\No \thickapprox y = \{\:\mid (\sqrt{x\cdot x+1_\No}-x )\cdot (\sqrt{x\cdot x+1_\No}-x )\}$
but $\sqrt{-1_\No} = -1_\No$ and $\sqrt{y} < -x$.
With this, for any negative number $x$ we can construct a number $\thickapprox -1_\No$,
whose square root is less than $x$.

\section{Reals and Ordinals as Subsets of Surreal}\label{s:subset}
Conway \cite{Conway} showed that the real numbers are a subset of $\No$ without focusing on their construction.
Starting with a construction of {$^\ast$integers} (that include the {$^\ast$naturals})
and inverse, an $x$ would be called {\it $^\ast$real} if
$x\thickapprox\{x-\frac{1_{\No}}{n_{\No}}\mid x+\frac{1_{\No}}{n_{\No}}\}_{0<n}$ and $-k_{\No}< x < k_{\No}$
for some natural $k$. This was later restricted to dyadic numbers.
Grimm  \cite{grimm2012} directly constructed the dyadic numbers and defined
a bijection from *real surreal into real.

\newcommand{\Sur}[0]{\mathrm{s}}

We formalize these constructions, additionally showing the set inclusions
(similar to the MML's property $\mathbb{N}\subseteq\mathbb{Z}\subseteq\mathbb{Q}\subseteq\mathbb{R}$ but with dyadic numbers $\mathbb{D}$ and ordinals):
{\setlength{\abovedisplayskip}{3pt}%
\setlength{\belowdisplayskip}{5pt}%
\begin{equation}\label{SDdef}
\strut\mkern-20mu\Sur_\mathbb{Z}(i) = \left\{
\begin{array}{cl}
0_{\No} & \mathrm{if}\;i=0, \\
\strut\mkern-2mu\{\Sur_\mathbb{Z}(i-1)\mid\,\} & \mathrm{if}\;i>0,\\
\strut\mkern-2mu\{\,\mid \Sur_\mathbb{Z}(i+1)\} & \mathrm{if}\;i<0,
\end{array}\right.
\Sur_\mathbb{D}(d) = \left\{
\begin{array}{cl}
\Sur_\mathbb{Z}(d) & \textrm{if}\;d \in\mathbb{Z}, \\
\strut\mkern-3mu\{\Sur_\mathbb{D}(\frac{j}{2^p})\mid \Sur_\mathbb{D}(\frac{j+1}{2^p})\}
& \textrm{if}\;d=\frac{2j+1}{2^{p+1}}\;\textrm{for}\\
& \textrm{some}\;j\in \mathbb{Z}, p \in \mathbb{N}.
\end{array}\right.
\end{equation}
}and construct $\Sur_\mathbb{R}(r)$
which selects the $\thickapprox$ equivalence class representative of the number, such that
$\{\Sur_\mathbb{D}(\frac{\lceil r\cdot(2^n)-1 \rceil}{2^n})\mid \Sur_\mathbb{D}(\frac {\lfloor r\cdot (2^n)+1 \rfloor}{2^n})\}_{n\in\mathbb{N}}$,
 $\restr{\Sur_\mathbb{R}}{\mathbb{D}} = \Sur_\mathbb{D}$, $\restr{\Sur_\mathbb{D}}{\mathbb{Z}} = \Sur_\mathbb{Z}$, and the basic operations $+$, $\cdot$ are preserved. We discuss
only two most important points: the set-theoretic definition of $\Sur_\mathbb{D}:\mathbb{D}\mapsto\Day{\omega}$
and the choice operator.

We denote the set of dyadic numbers of the form $\frac{j}{2^n}$ (where $j\in\mathbb{Z}$) as $\mathbb{D}_n$.
Observe that the sequence $\mathbb{D}_0=\mathbb{Z}$, $\mathbb{D}_n$ is increasing and
$\bigcup_{n \in \mathbb{N}} \mathbb{D}_n=\mathbb{D}$.
Then, for any $n$ we define a {\it recursive operator}
$I_n: (\Day{\omega \oplus n})^{\mathbb{D}_n}\mapsto (\Day{\omega\oplus (n\!+\!1)})^{\mathbb{D}_{n\!+\!1}}$
which extends the domain of $\mathbb{D}_n$ to $\mathbb{D}_{n\!+\!1}$,
assigning $\mathbb{D}_{n\!+\!1}\setminus \mathbb{D}_n$ values according to \eqref{SDdef}.
It easily follows that
$\{a\mid b\} \in \Day{\omega\oplus (n\!+\!1)}$ if $a,b\in \Day{\omega\oplus n}$.
Then using $\Sur_\mathbb{Z}$ on $\mathbb{D}_0$ with MML's fixed point combinator
\texttt{NAT\_1:sch 11}
we construct $\Sur_\mathbb{D}$ and prove that the values belong to
$\bigcup_{n \in \mathbb{N}}\Day{n}$.

We now want to obtain $\restr{\Sur_\mathbb{R}}{\mathbb{D}} = \Sur_\mathbb{D}$. To choose a representative of the equivalence class $[x]$ for a given $x$
we can use ``gluing'' $\Sur_\mathbb{R}(r) =\Sur_\mathbb{D}(r)$ for dyadic $r$ and use global choice otherwise (taking special care, since
choosing $c$ from $[x]_\thickapprox$ is a proper class).
Using an adaptation of {\it Scott's trick}, we can expect $c$ to be the youngest (in the sense of%
$\mathfrak{b}$).
We therefore can replace the proper class $[x]_\thickapprox$ by the set $\{y \in \Day{\born{x}}:y \thickapprox x\}$.
However, this set still can have several elements,
so in the first place we aim to reduce the cardinality of
$L_c$, $R_c$ and the cardinality of their union.
Using the Hessenberg sum, we can
minimize all three.
For this, we will
use  $\overline{\overline{L_c}} \oplus \overline{\overline{R_c}}$ instead of $\overline{\overline{L_c\cup R_c}}$
and use properties of $\thickapprox$.
Finally, we would like the globally selected $c$ to be suitable. By suitable,
we mean minimal w.r.t. $\born{}$ as well as having only suitable elements in $L_c, R_c$.
For this, we  again use a transfinite sequence, the last element of which has only suitably selected elements
of $\Day{\alpha}$%
{
\begin{Mizar}{x,Y,A}
  let'$\alpha$'be'Ordinal;
  func $\Uniq{}$_op'$\alpha$'->'Sequence'means'::'SURREALO:def 29
    dom'it'='succ'$\alpha$'$\wedge$'$\forall\,\beta$'be'Ordinal.'$\beta$'$\in$'succ'$\alpha$'$\Rightarrow$'(it.$\beta$'$\subseteq$'$\Day\beta$'$\land$
       ($\forall\,$x'be'Surreal.'x'$\in$'it.$\beta$'$\Leftrightarrow$'(x'$\in$'union'rng'($\restr{it}{\beta}$)'$\vee$'($\beta$'='born_eq'x'$\land$
         $\exists\,$Y'be'non'empty'surreal-membered'set'.
             Y'='born_eq_set'x'$\cap$'made_of'union'rng'($\restr{it}{\beta}$)'$\land$'x'='the'Y'-smallest'Surreal))));
\end{Mizar}
}
\noindent where \miz{the} is the global choice operator, \mizV[x]{born_eq x} is the $\mathfrak{b}$
of youngest surreal that is $\thickapprox x$, \mizV[x]{born_eq_set x} is the set of youngest surreal that is $\thickapprox x$,
\mizV[X]{made_of'X} is the set of such surreals that
their both left and right members belong to $X$, and
\mizV[Y]{Y-smallest} means that it has the smallest cardinality, that is
$\overline{\overline{L_c}} \oplus \overline{\overline{R_c}}$.

The definition implicitly assumes that \mizV[x]{born_eq_set'x'$\cap$'made_of'$\bigcup$'rng'($\restr{it}{\beta}$)} is non-empty because \miz{Y} is non-empty.
In consequence, \mizV[Y]{Y-smallest Surreal} is also non-empty
and we can use global choice.
Finally, we can use
transfinite induction to select a unique element for each class $[x]_\thickapprox$, denoted $\Uniq{x}$.

It is important to notice, that the properties that need to be proved about
the globally selected numbers
(such as youngest, smallest cardinality, suitable member) must be proved by
simultaneous induction, just like it was the case with the correctness of multiplication proofs in Section~\ref{s:field}.
We call the type
of such elements \miz{uSurreal}. This type is crucial for the definition of $\Sur_\mathbb{R}$, as values of $\Sur_\mathbb{D}$ are \miz{uSurreal},
and before $\Day{\omega}$ there are no more \miz{uSurreal}. This means that $\Sur_\mathbb{R}$ can be defined using $\Uniq{}$ without
``gluing''.

Next, we define {$^\ast${ordinal}} numbers to be all the surreal numbers $x$, for which $R_x=\emptyset$ (following Conway).
Subsequently, again using a transfinite sequence, we define the operator $\Sur_\On$ from ordinals to {$^\ast${ordinal}}.
We additionally apply $\Uniq{}$ to it, so that the values are \miz{uSurreal}
(this is justified, as we constructed $\Uniq{}$ to be {$^\ast${ordinal}} on {$^\ast${ordinal}}).
The proposed construction of \miz{uSurreal} builds a bridge that helps us formally combine proofs about the surreal numbers
in the general sense with the proofs that use the tree-theoretic definition that uses real and ordinal. This will be
important for several results in the next section.

\section{Conway Normal Form (CNF)}\label{s:cnf}
Conway has partitioned $\No$ using a property similar to Archimedianness. Under this partition, the surreals behave somewhat similarly to a vector space with arbitrary (ordinal-indexed) dimensions. The Conway Normal Form theorem will give a concept akin to coordinates in that space. These are all weak analogies, however, they give a hint as to why CNF is important for surreal numbers.

The most common ordered field, the real numbers, has the
Archimedean property, i.e., for all $x,y\in\mathbb{R}^+$, $x< n\cdot y$ for some $n\in\mathbb{N}$.
However, this is not true for infinite cardinals, so
the Archimedian-like partition of surreal is defined somewhat differently
(remember that $\Sur_\mathbb{Z}$ is the natural embedding of integers into surreal, as defined in the previous Section):

\begin{definition}[commensurate, infinitely less]
Let $x,y$ be surreal numbers. We say $x,y$ are {\it commensurate} if $x<\Sur_\mathbb{Z}(n)\cdot y \land y<\Sur_\mathbb{Z}(m)\cdot x$ for some $n,m\in \mathbb{N}$. We say $x$ is {\it infinitely less} than $y$
and write $x <^\infty y$ if $ x\cdot \Sur_\mathbb{Z}(n)< y$ for all
$n\in \mathbb{N}\setminus\{0\}$.
\end{definition}

As the definition of commensurate numbers only makes sense for positive numbers, we will refer to $x,y$ as \emph{commensurate in absolute terms} if $|x|$, $|y|$ are commensurate,
where $|\,\cdot\,|$ is the standard absolute value. (Similarly $<^\infty$ is defined in general, but only used for non-negative numbers.)
We next define the {\it power of $\Nomega$} (also called $\Nomega$-map~\cite{Conway})
as follows:
\begin{definition}[$\Nomega^\cdot$]
Let $x$ be a surreal. Then the $x$ {\it power of $\Nomega$}, written $\Nomega^x$ is defined as:
\begin{equation}
\Nomega^x = \{0_\No, \Sur_\mathbb{R}(r)\cdot\Nomega^{x_L} \mid \Sur_\mathbb{R}(s)\cdot \Nomega^{x_R}\}
\end{equation}
where $r,s$ range over all positive reals and $\Nomega =\Sur_{\On}(\omega)$.
\end{definition}
To define the function formally, we again use Theorem~\ref{seqexuni} with
$\theta = \born{x}$, $X(\alpha)=\Day\alpha$, and
\begin{multline}H(o,S) = \, \{ \{0_\No\} \cup \{(\bigcup rng\,S)(o^L)\cdot \Sur_\mathbb{R}(r) \mid o^L \in L_o \wedge r\in\mathbb{R}^+ \},\\
\{(\bigcup rng\,S)(o^R)\cdot \Sur_\mathbb{R}(s) \mid o^R \in R_o \wedge s\in\mathbb{R}^+ \}\}.\end{multline}
Note, that $\Nomega^x$ is different from exponentiation (the differences are subtle, see \cite[page 38]{Conway}),
but many of its properties are similar:
$\Nomega^0 = 1_\No$, $\Nomega^{x+y} = \Nomega^{x} \cdot \Nomega^{y}$, and
$\Nomega^{x} \leq \Nomega^{y}$ if $x \leq y$
and additionally $\Nomega^{x} <^\infty \Nomega^{y}$ if $x < y$.

The underlying idea for Conway's Normal Form of $x\not\thickapprox 0_\No$ is the observation that
using the power of $\Nomega$ we can determine a unique leader $\Nomega^y$ commensurate in absolute terms with $x$:

\begin{theorem}
Let $x\not \thickapprox 0_\No$. Then
there exists a unique $y$ being \miz{uSurreal} such that $r\in  \mathbb{R}\setminus\{0\}$
for which $|x - \Sur_\mathbb{R}(r)\cdot \Nomega^{y}|<^\infty \Nomega^{y}$.
\end{theorem}
Notice that when $x_1, x_2\not \thickapprox 0_\No$ have the same leader $\Nomega^y$ then $|x_1|$, $|x_2|$ are commensurate.
Additionally $\Nomega^y$ is the leader for $x_1\cdot\Sur_\mathbb{R}(r)$, $x_1+x_2$ for arbitrary $r\in\mathbb{R}\setminus\{0\}$.

\newcommand{\rseq}[0]{\texttt{r}}
\newcommand{\yseq}[0]{\texttt{y}}
\newcommand{\sseq}[0]{\texttt{s}}
\newcommand{\pseq}[0]{\texttt{p}}

Fix $x\not \thickapprox 0_\No$.
We can {\it approximate} $x$ using powers of
$\Nomega$, i.e., $x = \Sur_\mathbb{R}(r_0)\cdot \Nomega^{y_0}+x_1$ and $x_1 <^\infty \Nomega^{y_0}$.
Further, if $x_1\not \thickapprox 0_\No$, we obtain a better approximation by a sum:
$x = \Sur_\mathbb{R}(r_0)\cdot \Nomega^{y_0}+\Sur_\mathbb{R}(r_1)\cdot \Nomega^{y_1}+x_2$, where $x_2<^\infty \Nomega^{y_1} <^\infty \Nomega^{y_0}$
and so forth. As a consequence, $x$ can be represented in a form (already close to Cantor Normal Form):
\begin{equation}
x =
\Sur_\mathbb{R}(r_0)\cdot \Nomega^{y_0}+
\Sur_\mathbb{R}(r_1)\cdot \Nomega^{y_1}+
\Sur_\mathbb{R}(r_2)\cdot \Nomega^{y_2}+
\ldots
\Sur_\mathbb{R}(r_{k-1})\cdot \Nomega^{y_{k-1}}+x_k.
\end{equation}
Unfortunately, a finite number of iterations does not guarantee $x_k \thickapprox 0_\No$. As such, an infinite sum will
be necessary (infinite in the ordinal sense, so not just $\omega$).
To
state the CNF theorem, we must define this sum formally.
Conway assumes its existence,
but the formal proof of its convergence is actually more involved than the CNF proof. We adapt Erlich's approach \cite{Ehrlich}:
\begin{definition}[$\theta$-term]
Let $x$ be a surreal, $\alpha$, $\alpha^\prime$ be ordinals.
Let $\rseq\coloneqq\{\rseq_\beta\}_{\beta\leq \alpha}$ be a sequence of 
non-zero reals,
$\yseq\coloneqq\{\yseq_\beta\}_{\beta\leq \alpha}$ be a decreasing sequence of surreals, and
 $\sseq\coloneqq\{\sseq_\beta\}_{\beta\leq \alpha^\prime}$ be a sequence of surreals where $\alpha \leq \alpha^\prime$
 (i.e., $\sseq$ can be longer than $\rseq$, $\yseq$).
 We call $x$ the $(\theta,\sseq,\yseq,\rseq)$-term if
 $| x - (\sseq_\theta + \Sur_\mathbb{R}(\rseq_\theta)\cdot \Nomega^{\yseq_\theta} ) | <^\infty \Nomega^{\yseq_\theta}$
 where $\theta\leq \alpha$.
 Additionally, we write $x\in \bigcap_{\theta,\sseq,\yseq,\rseq}$ if $\theta\leq \alpha$ and $x$ is
 $(\gamma,\sseq,\yseq,\rseq)$-term for arbitrary $\gamma < \theta$.
\end{definition}
Note, that for any $a\leq b\leq c$, if $a \in \bigcap_{\theta,\sseq,\yseq,\rseq}$ and
$c \in \bigcap_{\theta,\sseq,\yseq,\rseq}$
then $b \in \bigcap_{\theta,\sseq,\yseq,\rseq}$.
Conway calls classes with this property {\it convex}.

Now we can formally express Conway's sentence
{\it the simplest number whose $\beta$-term is $\rseq_\beta\cdot\Nomega^{\yseq_\beta}$} in Erlich's way \cite{Ehrlich}.
We say that a triple $\sseq,\yseq,\rseq$ is simplest on $\beta$ if
\begin{itemize}
\item $\sseq_\beta=0_\No$ for $\beta=0$,
\item if $0 <\beta$ holds:
  $\sseq_\beta$ is \miz{uSurreal}, $\sseq_\beta \in\bigcap_{\beta, \sseq,\yseq,\rseq}$
   and for every \miz{uSurreal} $a\not=\sseq_\beta$, if $b \in \bigcap_{\beta,\sseq,\yseq,\rseq}$ then $\born{\sseq_{\beta}}<\born{a}$.
\end{itemize}
The expression $\bigcap_{\beta, \sseq,\yseq,\rseq}$ depends on all $\sseq_\gamma$ for
$\gamma\!<\!\beta$, so we can use it to specify $\sseq_\beta$.

Let $\theta$ be an ordinal. Consider, as previously, $\{\rseq_\beta\}_{\beta<\theta}$,
$\{\yseq_\beta\}_{\beta<\theta}$
and let $\{\sseq_\beta\}_{\beta\leq\theta}$ be a sequence of \miz{uSurreal}
where additionally the triple $\sseq,\yseq,\rseq$ is simplest on $\beta$
for $\beta\leq \theta$.
Then $\sum_{\beta<\theta} \Nomega^{\yseq_\beta} \cdot \rseq_\beta$ is defined to be $\sseq_\theta$ \cite{Ehrlich,Conway}.
As usual, to construct a suitable $\theta$-long sequence $\sseq$
we apply transfinite induction, first showing the existence of a suitable $\beta$-long sequence
for $\beta\leq \theta$.
Indeed,
using the $\beta$-step assumption
we can construct a suitable sequence
$\pseq:=\{p_\gamma\}_{\gamma < \beta}$ and extend it by the
assignment $p_\beta = \Uniq{e}$ for some $e\in\bigcap_{\beta,\pseq,\yseq,\rseq}$.  The existence of such $e$ is the key problem:
If $\beta$ is a limit ordinal, i.e., the sequence $\{p_\gamma\}_{\gamma < \beta}$
does not have a last element\footnote{We focus on the case of $\beta$ being the limit ordinal, where $\gamma < \beta\: \Longleftrightarrow\: 1\!+\!\gamma <\beta$.}.
 Conway introduced it highly informally \cite{Conway}.
Erlich \cite{Ehrlich} does it formally
using an approach where the class of every $e$ where $e \in\bigcap_{\beta,\pseq,\yseq,\rseq}$
corresponds to a non-empty intersection of the
descending transfinite sequence of convex subclasses of $\No$.
This assumes a stronger foundation and is not possible in Mizar (nor Isabelle/ZF or Metamath).
Indeed, his surreal numbers with lexicographic order are full, equivalently complete or equivalently
every nested sequence  $\{I_\gamma\}_{\gamma<\beta}$ of non-empty convex subclasses has a non-empty intersection
(see Theorem 4 in \cite{Ehrlich}). We cannot do this
for $\gamma<\beta$ when $\beta$ is a limit ordinal.
To solve this in standard set theory,
we defined two
somewhat complicated
sequences $\{l_\gamma\}_{\gamma<\beta}$, $\{u_\gamma\}_{\gamma<\beta}$, defined by,
$l_\gamma := p_{1\!+\!\gamma} +(\Sur_\mathbb{R}(\rseq_{1\!+\!\gamma})-1_\No)\cdot \Nomega^{\yseq_{1\!+\!\gamma}}$,
$u_\gamma := p_{1+\gamma} +(\Sur_\mathbb{R}(\rseq_{1+\gamma})+1_\No)\cdot \Nomega^{\yseq_{1\!+\!\gamma}}$ for $\gamma<\beta$
that in contrast to $\pseq$
are monotonous
(increasing and decreasing, respectively)
and $\{\bigcup_{\gamma<\beta}\{l_\beta\}\mid \bigcup_{\gamma<\beta}\{u_\beta\}\}$ is a member of the intersection.

With these, the formalization of the following theorem is attainable.
It says that the approximation of $x$ in $\Nomega$ way can be performed at most
$\born{x}$ times, since their $\mathfrak{b}$-s of subsequent partial approximating sums
give an increasing sequence bounded by $\born{x}$.

\begin{theorem}[Conway's Normal Form Theorem, Mizar ID: SURREALC:100,102]\label{CNFT}
For every surreal $x$ there exists a unique
$\{\rseq_\beta\}_{\beta < \theta}$ sequence of non-zero real,
$\{\yseq_\beta\}_{\beta< \theta}$ decreasing sequence of \miz{uSurreal} such that
  $x\thickapprox \sum_{\beta<\theta} \Nomega^{\yseq_\beta} \cdot \Sur_\mathbb{R}(\rseq_\beta)$.
Moreover $\mathfrak{b}$ of the sum  $\le$ $\born{x}$.
\end{theorem}

Conway's Normal Form allows characterizing any $x$ using a sum of two transfinite sequences $\rseq$, $\yseq$.
Rather than a regular sum, it is interpreted more as a Hahn-Mal'cev-Neumann infinite series sum.
This characterization is key to the further formalization of Conway \cite{Conway}, in particular it will allow
constructing the $n^{th}$-root of $x$, showing that odd-degree polynomials have roots, characterizing
omnific surreal integers and further results as discussed in the conclusion, Section~\ref{s:concl}.

\section{Related Work}
There are several formalization pertinent to surreal numbers in various systems.
Mamane's formalization in Coq \cite{Mamane04}, Obua's in Isabelle/HOLZF \cite{holzf}, and
Carneiro, Morrison, and Nakade's\footnote{\url{https://github.com/leanprover-community/mathlib/blob/master/src/set_theory/surreal/}} in Lean
follow an approach closer to Conway. All three avoid induction-recursion by starting with games.
Mamane motivates his work as a ``stress-test'' of Coq in the formalization of a very set-theoretic definition.
He proved that surreal numbers form a commutative ring (without associativity),
and without {\it permuting induction}, that was only formulated in the article.
Without this induction, the formalization needs to cover $2^n$
cases corresponding to the edges of an $n$-dimensional cube. Our work deal with this
using $\OProd{}$, $\CProd{}$ inductively to cover the cartesian product.

Obua's work focues on the development of the infrastructure for surreal numbers.
The formalization only reachers the fact that surreals form an additive group.
Similarly, the Lean formalization defines surreal numbers with addition and show that they form a commutative group.
It also includes an embedding of ordinals into surreal, a manually defined halving operator and
an embedding of dyadic numbers into surreal.

Induction-recursion, as studied in intuitionistic type theory \cite{Dybjer00}, could allow a more direct definition of the
surreal numbers. However, we are not aware of any formalizations of the surreal numbers that make use of induction-recursion.

Nittka followed the tree-theoretic approach in Mizar \cite{CGAMES_1.ABS}. After showing the involutiveness of minus,
further definitions and properties became too involved in this approach in Mizar and made us abandon this method.
With the \miz{uSurreal} obtained in our formalization, it is possible to easier continue with that approach.

The largest formalization of surreal numbers focusing solely on the
tree-theoretic approach has been developed in the Megalodon proof
system\footnote{http://grid01.ciirc.cvut.cz/~chad/100thms/100thms.html}.
Without the use of permuting induction, it shows that surreals form a
field and defines the square root.  Megalodon stands out as the only
system where integers, reals and ordinals used throughout the system
are carved out of the surreal numbers, rather than being added on
top. In comparison to our work, the Megalodon formalization lacks the
Conway Normal Form theorem and the theorems and definitions leading up
to it. Additionally, we formalized morphisms between the MML numbers
and Conway numbers, enabling transfer of theorems between them. The
Megalodon formalization of the inverse and square root operations has
been completed based on the ones done in our Mizar formalizations,
demonstrating the adaptability of our approach to other systems.

\section{Conclusion}\label{s:concl}
We formalized a large number of properties of surreal numbers in the
Mizar proof assistant system.  We initially focused on Conway's
approach to introduce the concept, which simplifies the definitions of
arithmetic operations, and then showed the equivalence of our approach
to the tree-theoretic approach. For this, we built a bridge that
allows joining the proofs in both approaches (\miz{uSurreal}) and
using it, we were able to reach Conway's Normal Form.  Due to the
relatively weak foundations of Mizar (first, there is no
induction-recursion; second, reasoning must be explicitly conducted on
sets rather than classes), we believe that our approach can be useful
for other formal systems.

CNF is crucial for further formalization of Conway's results
\cite{Conway}. Future work includes a formalization of $n^{th}$-root
of $x$. We are considering two approaches. First, to combine the use
of Kruskal-Gonshor exponential function with
logarithms. Alternatively, a more direct use of CNF, following
\cite{Conway}, is possible. With the $n^{th}$-root of $x$, one can
show that $\No$ is algebraically real-closed, i.e., that odd-degree
polynomials have roots. CNF is also needed to characterize omnific
surreal integers (i.e., surreals that satisfy
$x\thickapprox\{x-1\mid x+1\}$). It is then possible to show that
every surreal number can be represented as the quotient of two omnific
integers. This can lead to the formalization of surcomplexes.  Finally,
CNF is fundamental for further works on s-hierarchical ordered field
$\No$ \cite{Ehrlich2012,Ehrlich}.



\end{document}